\documentclass[journal]{IEEEtran}

\usepackage[latin1]{inputenc}
\usepackage{graphicx, amssymb, amsmath, amsfonts, amsthm}
\usepackage[usenames,dvipsnames]{color}

\allowdisplaybreaks
\IEEEoverridecommandlockouts

\newtheorem{theorem}{Theorem}
\newcommand{\PP}{\mathbb{P}}
\newcommand{\E}{\mathbb{E}}
\newcommand{\s}{s_{\rm th}}
\newcommand{\lcr}{{\rm LCR}_{S(t)}(\s)}
\newcommand{\afd}{{\rm AFD}_{S(t)}(\s)}

\makeatletter
\newlength \figwidth
\makeatother

\begin{document}
\title{Continuous Fluid Antenna Systems:\\ Modeling and Analysis}

\author{Constantinos Psomas, \IEEEmembership{Senior Member, IEEE}, Peter J. Smith, \IEEEmembership{Fellow, IEEE}, Himal A. Suraweera, \IEEEmembership{Senior Member, IEEE}, and Ioannis Krikidis, \IEEEmembership{Fellow, IEEE}\vspace{-9mm}
\thanks{Constantinos Psomas and Ioannis Krikidis are with the Department of Electrical and Computer Engineering, University of Cyprus, 1678 Nicosia, Cyprus (e-mail: psomas@ucy.ac.cy; krikidis@ucy.ac.cy).
	
Peter J. Smith is with the School of Mathematics and Statistics, Victoria University of Wellington, Wellington 6012, New Zealand	(e-mail: peter.smith@ecs.vuw.ac.nz).

Himal A. Suraweera is with the Department of Electrical and Electronic Engineering, University of Peradeniya, Peradeniya 20400, Sri Lanka (e-mail: himal@ee.pdn.ac.lk).

This work has received funding from the European Research Council (ERC) under the European Union's Horizon 2020 research and innovation programme (Grant agreement No. 819819). It was also co-funded by the European Regional Development Fund and the Republic of Cyprus through the Research and Innovation Foundation, under the project CONCEPT/0722/0123.}}

\maketitle

\begin{abstract}
Fluid antennas (FAs) is a promising technology for introducing flexibility and reconfigurability in wireless networks. Recent research efforts have highlighted the potential gains that can be achieved in comparison to conventional antennas. These works assume that the FA has a discrete number of positions that the liquid can take. However, from a practical standpoint, the liquid moves in a continuous fashion to any point inside the FA. In this paper, we focus on a continuous FA system (CFAS) and present a general framework for its design and analytical evaluation. In particular, we derive closed-form analytical expressions for the level crossing rate (LCR) and the average fade duration of the continuous signal-to-interference ratio (SIR) process over the FA's length. Then, by leveraging the LCR expression, we characterize the system's outage performance with a bound on the cumulative distribution function of the SIR's supremum. Our results confirm that the CFAS outperforms its discrete counterpart and thus provides the performance limits of FA-based systems.
\end{abstract}


\section{Introduction}
The rapid evolution of wireless networks has been achieved through numerous technological advancements such as multiple-input multiple-output (MIMO) systems \cite{HEATH}. Nonetheless, from a practical point-of-view, conventional antennas are rigid and inflexible as they are mostly made of metal. Moreover, they are developed to conform to a fixed configuration (e.g., specific frequency bands) and their use in devices is greatly limited by the available space. As such, the establishment of reconfigurable wireless networks has been of significant importance in recent years \cite{AEM}.

A promising approach are the so-called fluid antennas (FAs), which provide reconfigurability as well as flexibility at the radio frequency front-end \cite{JOM}. Specifically, the reconfiguration of FAs is achieved by displacing a metallic (or non-metallic) liquid (e.g., Mercury, Galinstan) inside a dielectric holder. This is done in a programmable and controllable manner through the use of microfluidic and electrowetting techniques. Therefore, with just a single structure, FAs can support multiple switchable configurations for different network requirements (e.g., changes in operating frequency or radiation pattern) \cite{JOM} or for dynamic environments (e.g., due to mobility or interference) \cite{KIT}. In view of this, the FA technology provides the potential to address fundamental design restrictions and push further the performance limits of wireless networks.

The performance of FAs in wireless communications has already been investigated for several scenarios \cite{KIT2}-\cite{BT}. The work in \cite{KIT2} focuses on the multiple access and demonstrates that significant capacity gains can be achieved as hundred of users, each with a single FA, can be supported. In \cite{KIT3}, it is shown that a single FA can outperform a conventional multi-antenna maximum ratio combining (MRC) system when the number of FA ports is large enough. This work is extended to consider ergodic capacity in \cite{KIT4} for the ergodic capacity where it is established that an FA can perform as well as an MRC system, even with a small size FA. The authors of \cite{CP} showed how the performance of FAs is negatively affected by outdated channel estimates and proposed prediction and coded modulation schemes to overcome these imperfections. The benefits provided by FAs in terms of physical layer security are demonstrated in \cite{BT}, where without any channel state information at the base station, the FA system surpasses conventional systems. Finally, the work in \cite{WKN} considers MIMO FA systems and proves that the diversity and multiplexing tradeoff achieved by MIMO FA systems exceeds that delivered by traditional MIMO. The aforementioned works assume that the FA can move the liquid to a discrete set of pre-defined positions, referred to as ports. However, in practice the liquid can be displaced at any position inside its available space due to its analog nature of operation.

As such, in this paper, we address this issue by providing an analytical framework for a continuous FA system (CFAS) in which the liquid's position follows a continuous process. We derive novel closed-form expressions for fundamental metrics such as the level crossing rate (LCR) and the average fade duration (AFD) of the continuous signal-to-interference ratio (SIR) process over the FA's length. Then, by utilizing the LCR expression, we characterize the outage performance with a bound on the cumulative distribution function (CDF) of the SIR's supremum. Computer simulations validate the LCR and AFD expressions but also show the accuracy of the CDF bound at the upper tail. Our results demonstrate the gains achieved by the CFAS in comparison to its discrete counterpart. The presented analytical framework provides a complete and general methodology for the design, modeling and analysis of practical FA-based systems.

\section{System Model}\label{sys_model}
Consider the downlink in a simple point-to-point topology consisting of a conventional single-antenna transmitter and a single-FA receiver. We assume that, apart from the desired signal, the receiver also experiences $N$ interfering signals. Moreover, all wireless links are assumed to exhibit Rayleigh fading. The receiver displaces the liquid at any position, say $t$, inside a dielectric holder of linear dimension and length $T$, $0\leq t\leq T$; we assume that the displacement can occur instantly and without delays \cite{KIT2, KIT3}. Let $g(t) \sim \mathcal{CN}(0,\beta_0)$ be the desired signal's channel coefficient at the $t$-th position and $h_i(t) \sim \mathcal{CN}(0,\beta_I)$, $i\in\{1,\dots,N\}$, the interfering channels, assumed to be independent and identically distributed. Then, the received signal at the $t$-th position is given by
\begin{align}
y(t) = g(t) x_0 + \sum_{i=1}^N h_i(t) x_i + n(t), ~0\leq t \leq T,
\end{align}
where $x_0$ is the transmitted data symbol with $\E[|x_0|^2] = 1$, $x_i$ is the data symbol from the $i$-th interferer with $\E[|x_i|^2] = 1$, and $n(t)$ is additive white Gaussian noise (AWGN) with zero mean and variance $\sigma^2$. Since the liquid can take any position inside the holder, the signals $y(t)$, $0\leq t \leq T$, are assumed to be spatially correlated. By considering two-dimensional isotropic scattering, the correlation between the received signals at two positions spaced by $\tau$ can be modeled by the well-known Jake's model as \cite{HEATH}
\begin{align}\label{correlation}
\rho(\tau) = J_0\Big(2\pi \frac{\tau}{\lambda}\Big),
\end{align}
where $\lambda$ is the signal's wavelength and $J_0(\cdot)$ is the zeroth order Bessel function of the first kind. Note that any correlation model of the form
\begin{align}\label{correlation2}
\rho(\tau) = 1 - b \tau^2 + o(\tau^2),	
\end{align}
for small $\tau$ and $b > 0$ could also be applied here, where $o(\cdot)$ is the little-o notation; in our case, based on the series representation of \eqref{correlation} \cite{ISG}, we have $b = \pi^2/\lambda^2$.

In this paper, we focus on an interference-limited network, i.e., the noise power is negligible compared to the interference. As such, we focus on the SIR at the $t$-th position of the FA receiver, which can be written as
\begin{align}
S(t) = \frac{|g(t)|^2}{I},
\end{align}
where $I = \sum_{i=1}^N |h_i(t)|^2$ is the aggregate interference power. Fig. \ref{fig1} depicts the considered linear FA as well as the SIR process over the liquid's position inside the holder. Therefore, to maximize its performance, the FA receiver will displace the liquid to some position, $t_0$, that maximizes the above ratio. In other words, by assuming full channel knowledge, the receiver displaces the liquid at the location that achieves
\begin{align}\label{sup}
S^* &= \sup_{0\leq t\leq T} \{ S(t) \} = \left[\sup_{0\leq t\leq T} \left\{\frac{|g(t)|}{\sqrt{I}}\right\}\right]^2,
\end{align}
where $\sup_{0\leq t\leq T} \{ S(t) \}$ denotes the supremum of $S(t)$.

\begin{figure}\centering
  \includegraphics[width=\figwidth]{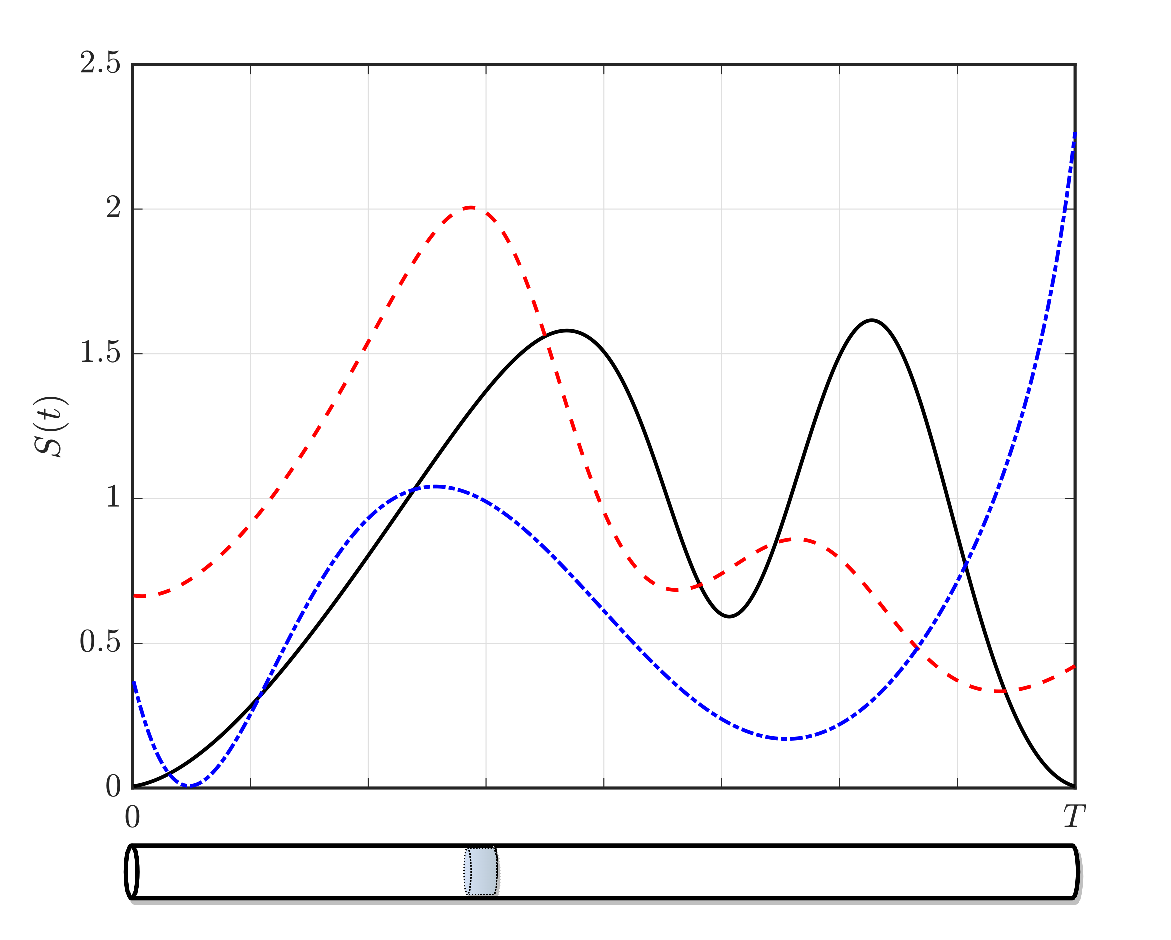}\vspace{-2mm}
  \caption{Three arbitrary scenarios of the SIR process over the position $t$.}\label{fig1}
\end{figure}

\section{Continuous SIR Processes in FAs}
To facilitate our analysis, we write \eqref{sup} as
\begin{align}
S^* &= \left[\sup_{0\leq t\leq T} \left\{\frac{|g(t)|}{\sqrt{I}}\right\}\right]^2 \triangleq \left[\sup_{0\leq t\leq T} \left\{\frac{\sqrt{\beta_0} r_0(t)}{\sqrt{\beta_I} r_I(t)}\right\}\right]^2,
\end{align}
with $\mathbb{E}[|g(t)|^2] = \beta_0$ and $\mathbb{E}[I] = \beta_I N$ so $r_0^2(t)$ is exponentially distributed with unit mean, i.e. $r_0^2(t) \sim \exp(1)$, and $r^2_I(t)$ is gamma distributed with shape parameter $N$ and unit scale parameter, i.e. $r^2_I(t) \sim {\rm \Gamma}(N,1)$; $\mathbb{E}[\cdot]$ denotes the expectation operator.

Now, ideally we would derive the distribution of $S^*$ but this seems to be challenging, if not impossible\footnote{Since the work in \cite{RD} from the mid-80s, there has been no further development on this kind of distribution. Moreover, in the communications literature, most works either deal with a small number of branches or with equally correlated branches, e.g. \cite{YC}.}. Instead, we focus on the LCR of $S(t)$ and leverage this metric to gain information about $S^*$. The LCR refers to the rate at which the SIR $S(t)$ crosses a level $\s$ in a positive (or negative) direction over the length of the FA and is given by \cite{NZ}
\begin{align}\label{lcr}
\lcr = \int_0^\infty \dot{s} f_{S,\dot{S}}(\s,\dot{s}) d\dot{s},
\end{align}
where $f_{S,\dot{S}}(\cdot,\cdot)$ is the joint probability distribution function (PDF) of $S$ and its derivative $\dot{S}$. It is important to point-out that the analytical expression for the LCR of such SIR processes is novel and has not appeared before in the literature.

\begin{theorem}\label{lcr_thm}
The LCR of a CFAS's SIR process over a threshold $\s$ is given by
\begin{align}
\lcr = \frac{\beta_0^N\Gamma(N+\frac{1}{2})}{\lambda\Gamma(N)(\beta_0+\beta_I \s)^N}\sqrt{\frac{2\pi\beta_I \s}{\beta_0}},
\end{align}
where $\Gamma(\cdot)$ denotes the gamma function.
\end{theorem}

\begin{proof}
See Appendix \ref{prf_lcr_thm}.
\end{proof}

Using the LCR, we can immediately obtain the AFD, defined by \cite{NZ}
\begin{align}\label{afd}
\afd = \frac{\mathbb{P}(S(t) < \s)}{\lcr},
\end{align}
where $\mathbb{P}(S(t) < \s)$ is the CDF of $S(t)$. The AFD is the average time that $S(t)$ remains below a level $\s$ after a downcrossing.

\begin{theorem}\label{afd_thm}
The AFD for a CFAS's SIR process is given by
\begin{align}
\afd &= \frac{\lambda\Gamma(N)}{\Gamma(N+\frac{1}{2})}\frac{(\s\beta_I+\beta_0)^N - \beta_0^N}{\sqrt{2\pi\beta_I \beta_0^{2N-1} \s}}.
\end{align}
\end{theorem}

\begin{proof}
See Appendix \ref{prf_afd_thm}.
\end{proof}

For the special case $\beta_0 = \beta_I = \beta$, the LCR and AFD are reduced to
\begin{align}
\lcr &= \frac{\Gamma(N+\frac{1}{2})}{\lambda\Gamma(N)(1+\s)^N} \sqrt{2\pi \s},\\
\afd &= \frac{\lambda\Gamma(N)}{\Gamma(N+\frac{1}{2})}\frac{(\s+1)^N-1}{\sqrt{2\pi \s}}.
\end{align}
We also look at two cases of interest concerning extreme values of the threshold $\s$. For small values of $\s$, we have
\begin{align}
\lcr &\sim \frac{\Gamma(N+\frac{1}{2})}{\lambda\Gamma(N)}\sqrt{\frac{2\pi \beta_I \s}{\beta_0}},\\
\afd &\sim \lambda\frac{\Gamma(N+1)}{\Gamma(N+\frac{1}{2})} \sqrt{\frac{\s\beta_I}{2\pi\beta_0}}.
\end{align}
On the other hand, for large $\s$ values,
\begin{align}
\lcr &\sim \sqrt{2\pi}\frac{\Gamma(N+\frac{1}{2})}{\lambda\Gamma(N)}\left(\frac{\beta_0}{\beta_I \s}\right)^{N-\frac{1}{2}},\\
\afd &\sim \frac{\lambda\Gamma(N)}{\sqrt{2\pi} \Gamma(N+\frac{1}{2})} \left(\frac{\beta_I \s}{\beta_0}\right)^{N-\frac{1}{2}}.
\end{align}

Next, we focus on the performance of $S^*$ and provide a bound on its CDF, that is, the outage probability.
\begin{theorem}\label{cdf_thm}
The CDF of $S^*$ is lower bounded by
\begin{align}
F_{S^*}(\s) \geq 1 &- \frac{\beta_0^N}{(\s\beta_I+\beta_0)^N}\nonumber\\
&\qquad\times\left(1 + T\frac{\Gamma(N+\frac{1}{2})}{\lambda\Gamma(N)} \sqrt{\frac{2\pi\beta_I \s}{\beta_0}}\right).
\end{align}
\end{theorem}

\begin{figure}\centering
	\includegraphics[width=0.95\figwidth]{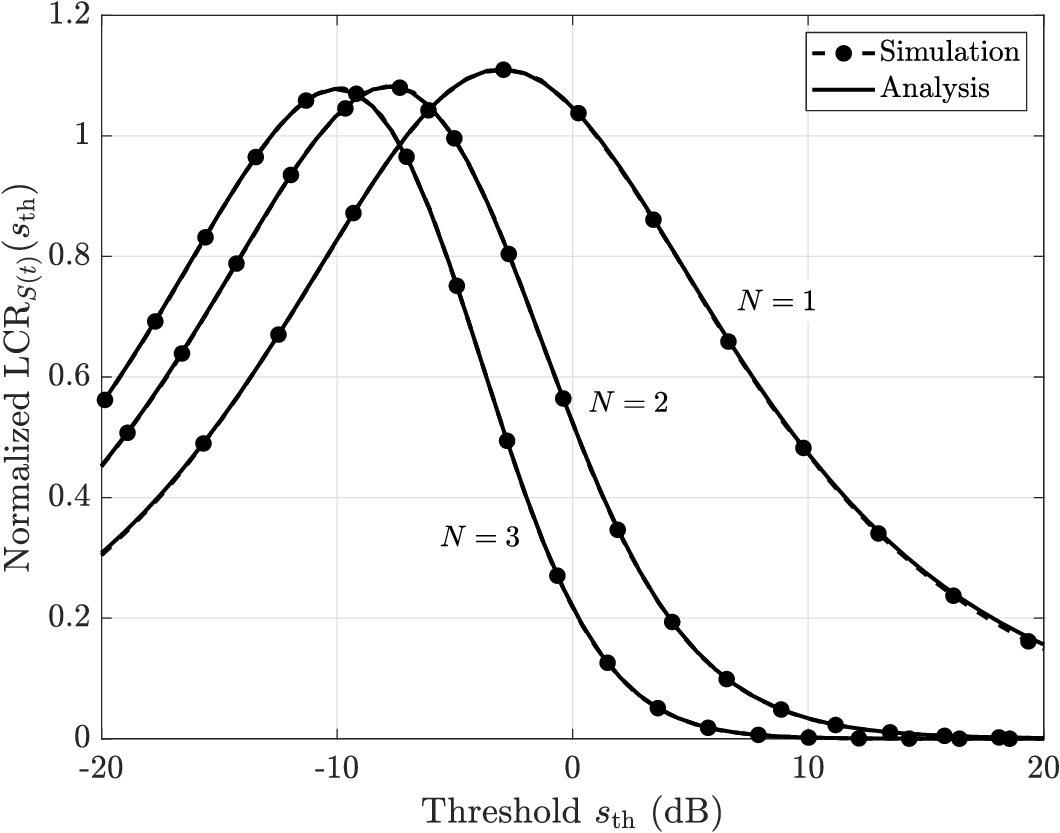}\vspace{-2mm}
	\caption{LCR with respect to threshold $\s$.}\label{fig2}\vspace{-1mm}
\end{figure}

\begin{proof}
See Appendix \ref{prf_cdf_thm}.
\end{proof}

Note that for very small values of $T$ we have
\begin{align}
	F_{S^*}(\s) \sim 1 - \frac{\beta_0^N}{(\s\beta_I+\beta_0)^N},
\end{align}
which corresponds to the CDF of $S(t)$, that is, the performance of displacing the liquid to a random position $t$, since the process $S(t)$ over a very small $T$ is strongly correlated. Moreover, Theorem \ref{cdf_thm} is asymptotically correct in the high upper tail ($\s \to \infty$) where we have
\begin{align}
F_{S^*}(\s) &\sim 1-\frac{\beta_0^N}{\s^N\beta_I^N}\Bigg(1 + T\frac{\Gamma(N+\frac{1}{2})}{\lambda\Gamma(N)}\sqrt{\frac{2\pi\beta_I \s}{\beta_0}}\Bigg)\nonumber\\
&= O(\s^{1/2-N}),
\end{align}
where $O(\cdot)$ is the big-O notation. The high upper tail region (high thresholds $\s$) mainly refers to the low transmit power regime, which is of particular interest for low-power energy efficient wireless networks such as the Internet-of-Things.

\begin{figure}\centering
	\includegraphics[width=0.95\figwidth]{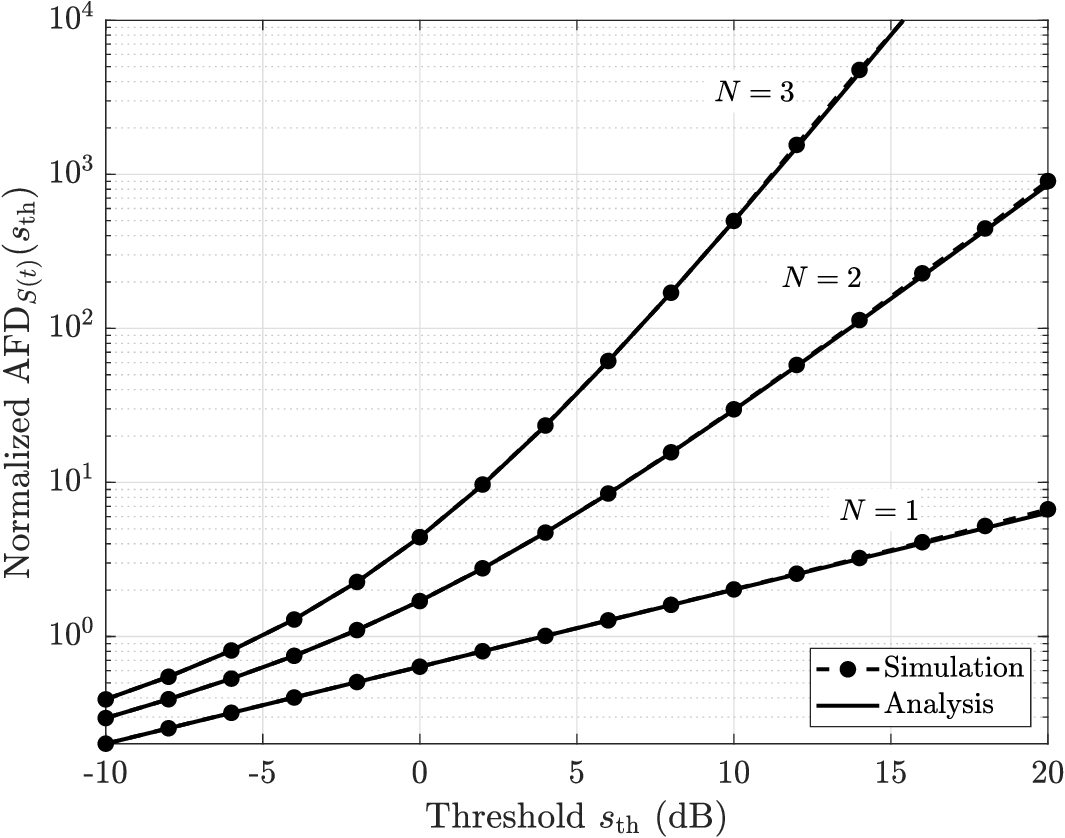}\vspace{-2mm}
	\caption{AFD with respect to threshold $\s$.}\label{fig3}\vspace{-1mm}
\end{figure}

\section{Numerical Results}
In this section, we validate our theoretical analysis with Monte Carlo simulations. For the sake of presentation, we consider $\beta_0 = 1$, $\beta_I = 2$, $\lambda = 1$ cm and up to three interferers, i.e. $N \in \{1,2,3\}$.

Fig. \ref{fig2} and Fig. \ref{fig3} depict the normalized LCR ($\lambda\lcr$) and the normalized AFD ($\afd/\lambda$), respectively, versus the threshold $\s$. As expected, the SIR crosses lower to medium threshold levels more frequently in the presence of more interferers. On the other hand, the SIR with $N=1$ achieves more frequent crosses at the medium to high thresholds regime compared to larger $N$. For the AFD, the SIR lies less time at lower threshold levels, whereas it stays longer at higher levels as the number of interferers increases. Finally, the simulation and analytical curves
perfectly match, which validates our theoretical analysis.

In Figs. \ref{fig4}(a) and \ref{fig4}(b), we show the CDF of $S^*$ for $T = 0.3$ cm and $T = 1$ cm, respectively. We can observe that a longer FA provides a better performance, as expected. Moreover, the significant performance losses as $N$ increases can be clearly seen. Note that the theoretical curve for the lower bound is accurate at the upper tail, which validates our claim in the previous section. What is more, the provided bound becomes tighter as $T$ decreases.

Finally, Fig. \ref{fig5} compares the CDF of a CFAS with the one obtained from a discrete FAS (DFAS) with $10$ ports. The CFAS outperforms the DFAS, in particular for medium to large threshold values, and thus provides	the performance limits of an FAS, as expected. Moreover, we illustrate the effect of the correlation model on the CDF. Towards this end, we consider the model for three-dimensional isotropic scattering, given by $\rho(\tau) = {\rm sinc} (2\pi \tau/\lambda)$ \cite{HEATH} but any other model satisfying \eqref{correlation2} can be easily adopted. In this case, the analysis follows the same steps but with $b = 2\pi^2/3\lambda^2$. This model increases the autocorrelation over the SIR process, which explains the gap between the two curves.

\begin{figure}\centering
	\includegraphics[width=0.95\figwidth]{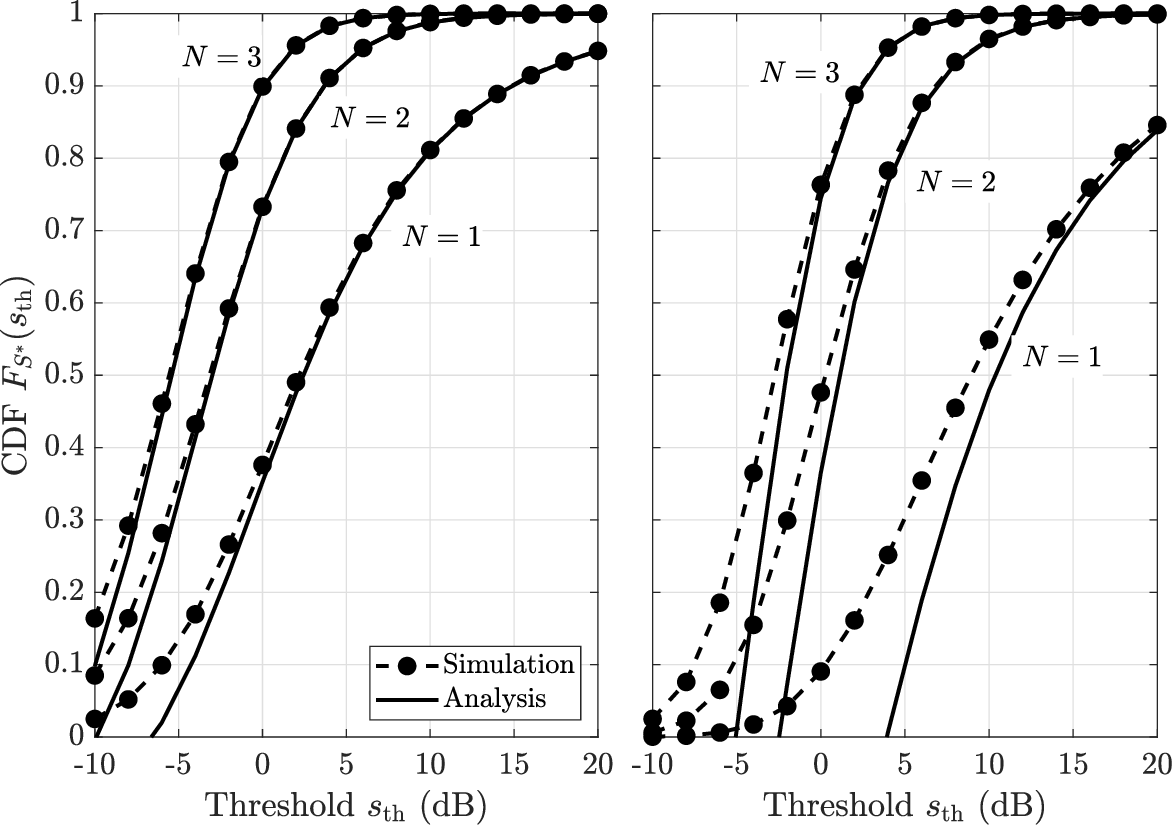}\\\hspace{2mm}
{\footnotesize (a) $T=0.3$ cm\hspace{16.5mm} (b) $T=1$ cm}\vspace{-1mm}
	\caption{CDF versus threshold $\s$.}\label{fig4}
\end{figure}

\begin{figure}\centering
	\includegraphics[width=0.95\figwidth]{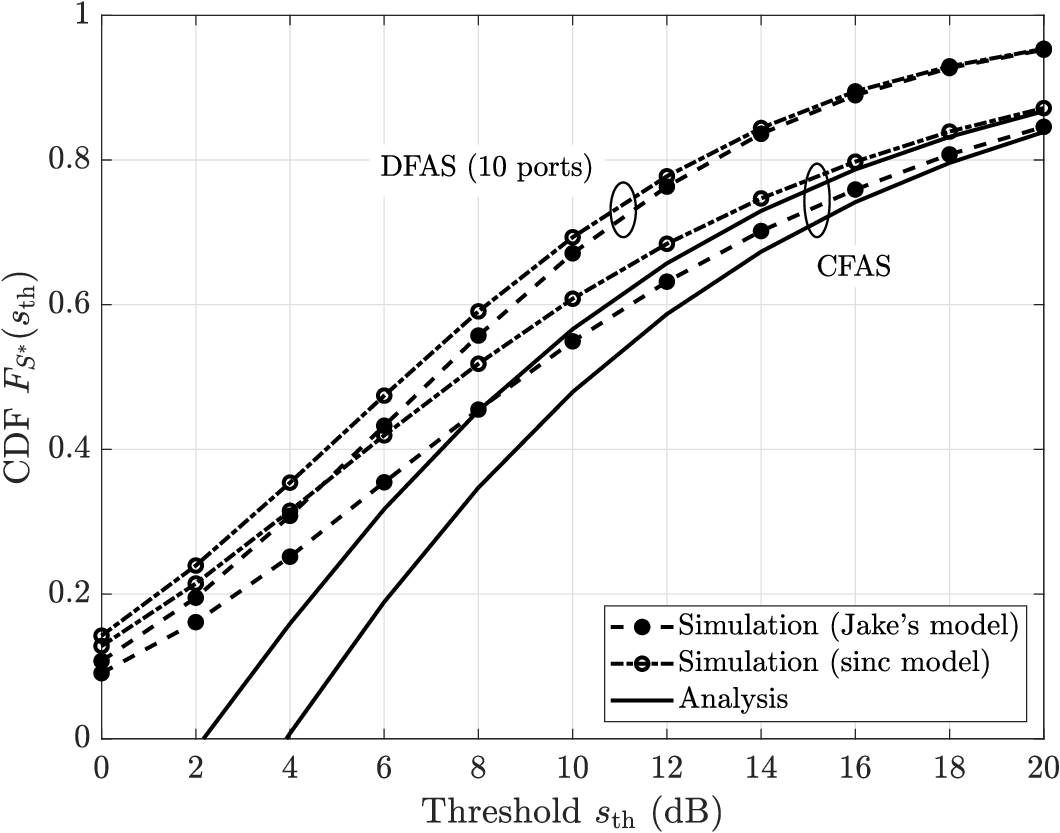}\vspace{-1mm}
	\caption{Comparison with the discrete case; $T=1$ cm and $N = 1$.}\label{fig5}
\end{figure}

\section{Conclusion}
In this paper, we considered a CFAS in which the liquid's position follows a continuous process. We derived closed-form expressions for the SIR's LCR and the AFD as well as a lower bound on the CDF of the SIR's supremum. We validated our analysis with computer simulations, which confirmed the accuracy of the LCR and AFD expressions but also the effectiveness of the CDF bound at the upper tail. Our results demonstrate the performance limits of an FA-based system. The provided framework is general and is suitable for designing and evaluating practical FA-based wireless systems. Future directions include a tigher bound for the lower tail or an exact expression for the SIR's supremum CDF.

\appendix
\subsection{Proof of Theorem \ref{lcr_thm}}\label{prf_lcr_thm}
It is known that the time derivative $\dot{r}$ of a Nakagami distributed $r(t)$ is independent of $r(t)$ and follows a Gaussian distribution with zero mean and variance $ b\Omega/m$, where $m$ and $\Omega$ are the shape and spread parameter, respectively \cite{NZ}, and $b$ is taken from \eqref{correlation2}. In our case, $b = \pi^2/\lambda^2$ and $\Omega = m$, so the derivatives of both $r_0(t)$ and $r_I(t)$ have variance $\pi^2/\lambda^2$. Now, let\vspace{-2mm}
\begin{align}
R(t) = \frac{r_0(t)}{r_I(t)} \implies \dot{R} = \frac{\dot{r}_0 r_I - \dot{r}_I r_0}{r_I^2}.
\end{align}
Then, as in \eqref{lcr}, the LCR of $R(t)$ across a threshold $r_{\rm th}$ is
\begin{align}
{\rm LCR}_{R(t)}(r_{\rm th}) = \int_0^\infty \dot{r} f_{R,\dot{R}}(r_{\rm th},\dot{r}) d\dot{r},
\end{align}
where $f_{R,\dot{R}}(\cdot,\cdot)$ is the joint PDF of $R$ and $\dot{R}$. We can write
\begin{align}
f_{R,\dot{R}}(r_{\rm th},\dot{r}) = \int_0^\infty f_{\dot{R}|R, r_I}(\dot{r}|r_{\rm th},y) f_{R,r_I}(r_{\rm th},y) dy,
\end{align}
and hence\vspace{-1mm}
\begin{align}
{\rm LCR}_{R(t)}(r_{\rm th}) = \int_0^\infty \int_0^\infty \dot{r}& f_{\dot{R}|R, r_I}(\dot{r}|r_{\rm th},y)\nonumber\\
&\times f_{R,r_I}(r_{\rm th},y) dy d\dot{r}.
\end{align}
Conditioned on $R = r_{\rm th}$ and $r_I = y$, we have that
\begin{align}
\dot{R} = \frac{\dot{r}_0}{y} - \frac{\dot{r}_I r_{\rm th}}{y} \sim \mathcal{N}\left(0, \frac{b}{y^2} + \frac{b r_{\rm th}^2}{y^2}\right),
\end{align}
with $b = \pi^2/\lambda^2$. Therefore, by using
\begin{align}
f_{\dot{R}|R,r_I}(\dot{r}|r_{\rm th},y) = \frac{1}{\sqrt{2\pi \left(\frac{b + b r_{\rm th}^2}{y^2}\right)}} \exp\left(-\frac{\dot{r}^2 y^2}{2(b + b r_{\rm th}^2)}\right),
\end{align}
and the fact that $\int_0^\infty x\exp(-x^2/a) dx = a/2$, we have
\begin{align}
{\rm LCR}_{R(t)}(r_{\rm th}) &= \int_0^\infty \frac{b + b r_{\rm th}^2}{y^2 \sqrt{2\pi \left(\frac{b + b r_{\rm th}^2}{y^2}\right)}} f_{R,r_I}(r_{\rm th},y) dy\nonumber\\
&=\sqrt{\frac{b + b r_{\rm th}^2}{2\pi}} \int_0^\infty \frac{1}{y} f_{R,r_I}(r_{\rm th},y) dy.\label{lcr1}
\end{align}
Now, by using standard transformation theory, we go from $(X,Y) = (r_0^2, r_I^2)$ to $(U,V) = (r_0/r_I, r_I)$. The forward transformations are $V = \sqrt{Y}$, $U = \sqrt{X/Y}$ and the inverse transformations are $Y=V^2$, $X=U^2V^2$. Thus, the Jacobian matrix is given by\vspace{-1mm}
\begin{align}
\mathbf{J} =
\begin{pmatrix}
2 u v^2 & 0\\
2 v u^2 & 2v
\end{pmatrix}.
\end{align}
Hence, we have\vspace{-1mm}
\begin{align}
&f_{U,V}(u,v) = f_X(u^2v^2) f_Y(v^2) 4uv^3\nonumber\\
&= \exp(-u^2v^2) \frac{1}{\Gamma(N)}v^{2(N-1)}\exp(-v^2)4uv^3,
\end{align}
since $X \sim \exp(1)$ and $Y \sim {\rm \Gamma}(N,1)$. By substituting the variable names $(U,V) = (r_0/r_I, r_I)$, we have
\begin{align}
&f_{R, r_I}(r_{\rm th},y) = \frac{4r_{\rm th}}{\Gamma(N)}\exp(-y^2(r_{\rm th}^2+1)) y^{2N+1}.\label{lcr2}
\end{align}
Then, by substituting \eqref{lcr2} in \eqref{lcr1}, we obtain
\begin{align}
{\rm LCR}_{R(t)}(r_{\rm th}) &= \sqrt{\frac{b+b r_{\rm th}^2}{2\pi}}\nonumber\\
&\quad\times \int_0^\infty \frac{4r_{\rm th}}{\Gamma(N)} y^{2N} \exp(-y^2(r_{\rm th}^2+1)) dy.
\end{align}
We know that $\int_0^\infty y^n \exp(-a y^2) = \frac{a^{-(n+1)/2}}{2} \Gamma\left(\frac{n+1}{2}\right)$ \cite{ISG}, and so
\begin{align}
\lcr_{R(t)}(r_{\rm th}) &= \sqrt{\frac{b+b r_{\rm th}^2}{2\pi}} \frac{4r_{\rm th}\Gamma\left(N+\frac{1}{2}\right)}{2\Gamma(N)(r_{\rm th}^2+1)^{N+1/2}}\nonumber\\
&= \sqrt{\frac{b+b r_{\rm th}^2}{2\pi}}\frac{2r_{\rm th}\Gamma\left(N+\frac{1}{2}\right)}{\Gamma(N)(r_{\rm th}^2+1)^{N+1/2}}.
\end{align}
Converting back to the original process
\begin{align}
\lcr = {\rm LCR}_{R(t)}(\sqrt{\beta_I \s/\beta_0}),
\end{align}
since $S(t) = \beta_0 R(t)^2/\beta_I$. Therefore,
\begin{align}
\lcr &= \sqrt{\frac{b+b \beta_I \s/\beta_0}{2\pi}}\nonumber\\
&\quad\times\frac{2 \sqrt{\beta_I \s/\beta_0}\Gamma\left(N+\frac{1}{2}\right)}{\Gamma(N)(\beta_I \s/\beta_0+1)^{N+1/2}},
\end{align}
and after several algebraic operations the theorem is proven.

\subsection{Proof of Theorem \ref{afd_thm}}\label{prf_afd_thm}
We first look at the CDF of $S(t)$, that is,
\begin{align}
\mathbb{P}(S(t) < \s) = \mathbb{P}\left(\frac{r_0^2(t)}{r_I^2(t)} < \frac{\beta_I}{\beta_0}\s\right).
\end{align}
The ratio $r_0^2(t)/r_I^2(t)$ has an $F$-distribution \cite{FD} with CDF
\begin{align}\label{cdf}
\mathbb{P}(S(t) < \s) &= 1-\frac{1}{(\s\beta_I/\beta_0+1)^N}\nonumber\\
&= 1- \left(\frac{\beta_0}{\s\beta_I+\beta_0}\right)^N.
\end{align}
As such, from \eqref{afd} and Theorem \ref{lcr_thm}, we have
\begin{align}
\afd &= \left[1-\left(\frac{\beta_0}{\s\beta_I+\beta_0}\right)^N\right] \sqrt{\frac{\beta_0}{2\pi\beta_I \s}}\nonumber\\
&\quad\times \frac{\lambda\Gamma(N)(\beta_0+\beta_I \s)^N}{\beta_0^N\Gamma(N+\frac{1}{2})},
\end{align}
and the result follows.

\subsection{Proof of Theorem \ref{cdf_thm}}\label{prf_cdf_thm}
We use an approach similar to \cite{MRL}. Specifically, we have
\begin{align}
\PP\left(S^* > \s\right) &= \PP\left(S^* > \s, S(0) < \s\right) \nonumber\\
&\qquad\qquad\qquad+\PP\left(S^* > \s, S(0) > \s\right)\nonumber\\
&= \PP(S^* > \s, S(0) < \s) + \PP(S(0) > \s),
\end{align}
where $S^*$ is the supremum of $S(t)$ and $S(0)$ is the SIR at the $0$-th position. Now, let $E$ be the event that at least one upcrossing across $\s$ occurs. Then,
\begin{align}
&\PP(S^* > \s, S(0) < \s) \leq \PP(E),
\end{align}
since
\begin{align}
&\{E\} = \{E \cap S(0) > \s\} \cup \{E \cap S(0) < \s\}\nonumber\\
&\qquad=\{E \cap S(0) > \s\} \cup \{S^* > \s, S(0) < \s\}.
\end{align}
Moreover,
\begin{align}
\PP(E) &= \PP({\rm one ~upcrossing})+\PP({\rm two ~upcrossings}) + \dots\nonumber\\
&\leq \E[{\rm no. ~of ~upcrossings}] = T\times \lcr.
\end{align}
Therefore, we have that
\begin{align}
\mathbb{P}(S^* > \s) \leq \mathbb{P}(S(0) > \s) + T\times\lcr,
\end{align}
where $\mathbb{P}(S(0) > \s)$ and $\lcr$ are given by \eqref{cdf} and Theorem \ref{lcr_thm}, respectively. Finally, the result follows from $F_{S^*}(\s) = 1-\mathbb{P}(S^* > \s).$

\end{document}